\documentclass[12pt,draftcls,onecolumn]{IEEEtran}
\usepackage{latexsym}
\usepackage{graphicx}
\usepackage{array}
\usepackage{amsmath}
\usepackage{amsfonts}
\usepackage{amssymb}
\usepackage{amsthm}

\newtheorem{thm}{Theorem}
\newtheorem{lemma}{Lemma}
\newtheorem{prop}{Proposition}
\newtheorem{defn}{Definition}

\newtheorem{remark}{Remark}
\newtheorem{cor}{Corollary}[thm]

\newcommand{\bX} {\boldsymbol{X}}

\newcommand{\bY} {\boldsymbol{Y}}
\newcommand{\bW} {\boldsymbol{W}}
\newcommand{\bF} {\boldsymbol{F}}
\newcommand{\bZ} {\boldsymbol{Z}}

\newcommand{\bXi} {\boldsymbol{\Xi}}

\newcommand{\sS} {\mathcal{S}}

\def\ba#1\ea{\begin{align}#1\end{align}}

\newcommand{{\bR}} {\right)}
\newcommand{\bp} {\begin{proof}}
\newcommand{\ep} {\end{proof}}
\newcommand{\bLF} {\left\{}
\newcommand{{\bRF}} {\right\}}

\newcommand{\uuline}[1]{\underline{\underline{#1}}}
\newcommand{\ooline}[1]{\overline{\overline{#1}}}

\newcommand{\oln} {\overline}
\newcommand{\uln} {\underline}

\begin{document}

\title{Feedback Does Not Increase the Capacity of Compound Channels with Additive Noise}

\author{Sergey Loyka,  Charalambos D. Charalambous




\thanks{The material in this paper was presented in part at the IEEE International Symposium on Information Theory, Barcelona, Spain, July 2016.}

\thanks{S. Loyka is with the School of Electrical Engineering and Computer Science, University of Ottawa, Ontario, Canada, e-mail: sergey.loyka@ieee.org}

\thanks{C.D. Charalambous is with the ECE Department, University of Cyprus, Nicosia, Cyprus, e-mail: chadcha@ucy.ac.cy}}

\maketitle


\begin{abstract}
A discrete compound channel with memory is considered, where no stationarity, ergodicity or information stability is required, and where the uncertainty set can be arbitrary. When the discrete noise is additive but otherwise arbitrary and there is no cost constraint on the input, it is shown that the causal feedback does not increase the capacity. This extends the earlier result obtained for general single-state channels with full transmitter (Tx) channel state information (CSI) to the compound setting. It is further shown that, for this compound setting and under a mild technical condition on the additive noise, the addition of the full Tx CSI does not increase the capacity either, so that the worst-case and compound channel capacities are the same. This can also be expressed as a saddle-point in the information-theoretic game between the transmitter (who selects the input distribution) and the nature (who selects the channel state), even though the objective function (the inf-information rate) is not convex/concave in the right way. Cases where the Tx CSI does increase the capacity are identified.

Conditions under which the strong converse holds for this channel are studied. The ergodic behaviour of the worst-case noise in otherwise information-unstable channel is shown to be both sufficient and necessary for the strong converse to hold, including feedback and no feedback cases.
\end{abstract}

\vspace*{-0.5\baselineskip}
\section{Introduction}

Many channels, especially wireless ones, are non-erogodic, non-stationary in nature \cite{Biglieri} so that the standard tools developed for stationary ergodic channels do not apply and new methods are needed for such channels. A powerful method to deal with general channels, for which stationarity, ergodicity or  information stability are not required, is the information density (spectrum) approach \cite{Verdu}\cite{Han}. In this method, the key quantity is the inf-information rate rather than the mutual information since the latter does not have operational meaning for information-unstable channels.

In real systems, channel state information (CSI) may be inaccurate or limited due to a number of reasons such as the limitations of channel estimation or feedback link \cite{Biglieri}. The concept of compound channel is one way to address this issue whereby a codebook is designed to work for any channel in the uncertainty set, without any knowledge of what channel state is currently in effect \cite{Lapidoth-98A}. A number of results have been obtained for the capacity of compound channels, see \cite{Lapidoth-98A} for a detailed review. While most of the studies do not consider feedback, the compound capacity of a class of finite-state memoryless (and hence information-stable) channels with deterministic feedback was established in \cite{Shrader-09}.

While most of the known results require some form of information stability for any channel in the uncertainty set, a general formula for compound channel capacity has been established in \cite{Loyka-15}\cite{Loyka-15-2} where no stationarity, ergodicity or information stability is required, and the uncertainty set can be arbitrary. The key quantity in this setting is the compound inf-information rate, which is an extension of the inf-information rate of \cite{Verdu}\cite{Han} to the compound setting.

In this paper, we extend the study in \cite{Loyka-15}\cite{Loyka-15-2} and consider a general compound channel with memory and additive noise (no information  stability is required so that the channel can be non-stationary, non-ergodic; the uncertainty set can be arbitrary), where all alphabets are discrete, there is no cost constraint and a noiseless, causal feedback link is present, where all past channel outputs are fed back to the transmitter. We consider a scenario where no CSI is available at the transmitter but full CSI is available to the receiver. Under this setting, we demonstrate that the feedback does not increase the compound channel capacity\footnote{In this paper, we consider the classical compound setting \cite{Biglieri}\cite{Lapidoth-98A}\cite{Shrader-09} where a fixed-rate code is designed to operate on any channel in the uncertainty set and its decoding regions are allowed to depend on the state (but not the encoding process); variable-rate coding, while being interesting, is beyond the paper's scope.}. This extends the earlier result in \cite{Alajaji-94} established for single-state fully-known channels (full CSI available at both ends) to the compound setting. Since noisy feedback cannot outperform noiseless one, this also holds for the former case.

Under a mild technical condition on the additive noise, we further show that the availability of the full Tx CSI does not increase the capacity either: the worst-case and compound channel capacities are the same. This fact is remarkable since achieving the worst-case capacity allows for the codebooks to depend on the channel state while the compound channel capacity requires the codebooks to be independent of the channel state (and hence no feedback to the Tx is needed). This can also be expressed as the existence of a saddle point in the information-theoretic game between the transmitter (who selects the input) and the nature (who selects the channel state): neither player can deviate from the optimal strategy without incurring a penalty. This result is rather surprising since the underlying objective function (the inf-information rate) is \textit{not} convex/concave in the right way and the uncertainty set can be non-convex as well (e.g. discrete) so that the celebrated von Neumann's mini-max Theorem \cite{Boyd-04} or its extensions \cite{Zeidler-86} cannot be used to establish the existence of a saddle-point. This shows that neither convexity of the feasible set nor of the objective function are necessary for a saddlepoint to exist in this information-theoretic game. This saddlepoint result extends the earlier results established for  stationary and ergodic (and hence information-stable) channels, e.g. in \cite{Dobrushin-59-2}-\cite{Loyka-15-3}, where mutual information is a proper metric, to the realm of information-unstable scenarios, where the inf-information rate has to be used as a metric since the mutual information does not have operational meaning anymore.

Next, we consider some cases when the Tx CSI does increase the capacity. This turns out to be somewhat surprising since, in all such cases, the optimal input distribution is uniform, regardless of the channel state (a common wisdom suggests that the Tx CSI increases the capacity via proper selection of the input distribution tailored to the channel state; our results indicate that this does not have to be the case). Examples are provided to facilitate understanding and insights.

Finally, conditions under which the strong converse holds for this channel are studied. The ergodic behaviour of the worst-case noise in otherwise information-unstable channel is shown to be both sufficient and necessary for the strong converse to hold, including feedback and no feedback cases. Examples are given to illustrate scenarios when the strong converse holds and when it does not.

The rest of the paper is organized as follows. Section II introduces the channel model and notations. Section III discusses the capacity of general (information-unstable) compound channels without feedback. The impact of feedback is included in Section IV. The impact of the channel state information at the transmitter and the existence of a saddlepoint are studied in Section V. Examples are given in Section VI. Sufficient and necessary conditions for the strong converse to hold are given in Section VII.

\section{Channel Model}

Let us consider the following discrete-time model of a compound discrete channel with additive noise:
\ba
\label{eq.ch.model}
y^n = g_s^n(x^n) + \xi^n_{s}
\ea
where $x^n, y^n, \xi^n_{s}$ are the input, output and noise sequences of length $n$, $x^n=\{x_1,..,x_n\}$ and likewise for $y^n$ and $\xi^n_s$; the functions $g_s^n(x^n)=\{g_{s1}(x_1), g_{s2}(x^2),.., g_{sn}(x^n)\}$ model the channel's impulse response and are required to induce one-to-one mapping $x^n \leftrightarrow z^n = g_s^n(x^n)$. All alphabets as well as operations are $M$-ary, $s \in \mathcal{S}$ denotes the channel (noise) state, and $\mathcal{S}$ is the (arbitrary) channel uncertainty set. The compound sequence $\xi^n_s = \{\xi_{1s},..,\xi_{ns}\}$ represents arbitrary additive noise, e.g. non-ergodic, non-stationary in general, independent of the channel input when used without feedback.

Note that the channel is not memoryless, it may include inter-symbol interference (ISI) via $g_s^n(\cdot)$, e.g.
\ba
\label{eq.ch.model.z}
z_k = g_{sk}(x^k) = \sum_{i=0}^{l_s} x_{k-i}
\ea
where $l_s$ is the depth of the ISI and where we set $x_i=0$ for $i<0$. The noise is also allowed to have arbitrary memory.

The channel is not required to be information stable (in the sense of Dobrushin \cite{Dobrushin-59} or Pinsker \cite{Pinsker-64}). We assume that $s$ is known to the receiver but not the transmitter, who knows the (arbitrary) uncertainty set $\mathcal{S}$. This is motivated by the fact that channel estimation is done at the receiver; $M$ may be small, e.g. binary alphabets, while the cardinality of $\mathcal{S}$ can be very large (in fact, $\mathcal{S}$ can be a continuous set) so it is not feasible in practice to feed $s$ back to the transmitter via e.g. a binary feedback channel.

The channel has noiseless feedback with 1-symbol delay (which can also be extended to noisy feedback - see Remark 3), so that the transmitted symbol $x_k$ at time $k=1..n$ is selected as $x_k^{(n)} = f_k^{(n)}(w^n y^{k-1})$\footnote{our result will also hold for a more general feedback of the form $u_k=\beta_k(y^k)$, where $\{\beta_k\}$ are arbitrary feedback functions, see Remark 4.}, where $n$ is the blocklength, $w^n$ denotes the message to be transmitted via $n$-symbol block, $x^n = \{x_1^{(n)} ... x_n^{(n)}\}$ and likewise for $y^n$ and $\xi_s^{n}$; $f_k^{(n)}$ denotes the encoding function at time $k$, which depends on the selected message $w^n$ and past channel outputs $y^{k-1}$ (due to the feedback); $f^n = \{f_1^{(n)} ... f_n^{(n)}\}$. This induces the input distribution of the form:
\ba
p(x^n||y^{n-1}) = \prod_{k=1}^n p(x_k|x^{k-1} y^{k-1})
\ea
where $||$ denotes causal conditioning \cite{Kramer-03}. No cost constraint is imposed on the input.

\textit{Notations}: To simplify notations, we use $p(x|y)$ to denote conditional distribution $p_{x|y}(x|y)$ when this causes no confusion (and likewise for joint and marginal distributions) and shortcut $x_k^{(n)}$ as $x_k$ with understanding that all sequences and distributions depend on blocklength $n$ and may be different for different blocklengths. Capitals ($X$) denote random variables while lower-case letters ($x$) denote their realizations or arguments of functions; $\bX = \{X^n\}_{n=1}^{\infty}$.

\begin{figure}[t]
\centerline{\includegraphics[width=3.5in]{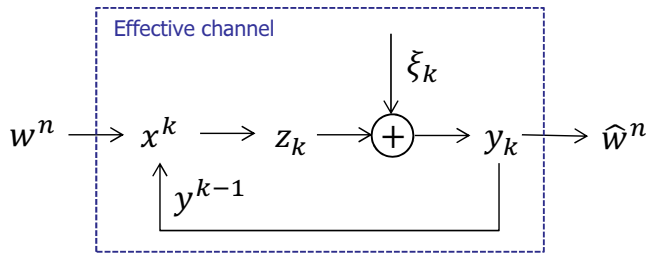}}
\caption{A general channel with additive noise and causal feedback $\{w^n y^{k-1}\} \rightarrow x^k \rightarrow y_k$ and the effective channel $w^n \rightarrow y^n$ (dashed box), where $z_k = g_{sk}(x^k)$, $x^k = f^k(w^n y^{k-1})$.}
\label{fig.1}
\end{figure}

\section{Capacity without Feedback}

First, we briefly review the relevant capacity result in \cite{Loyka-15}\cite{Loyka-15-2}, which apply to general compound channels $p_s(y^n|x^n)$, not only those in \eqref{eq.ch.model}; channels can be information-unstable, e.g. non-stationary, non-ergodic, but without feedback, i.e. $x^k = f^k(w^n)$ (the input depends only on the message and the past inputs, not the outputs). The compound channel capacity is defined operationally as the maximum achievable rate for which the error probability can be made arbitrary small and uniformly so over the whole set of channels and where the codewords are independent of channel state (see e.g. \cite{Lapidoth-98A}\cite{Loyka-15-2} for details).

\begin{thm}[\cite{Loyka-15}\cite{Loyka-15-2}]
\label{thm.C.general}
Consider a general compound channel where the channel state $s \in \mathcal{S}$ is known to the receiver but not the transmitter and is independent of the channel input; the transmitter knows the (arbitrary) uncertainty set $\mathcal{S}$. Its compound channel capacity (without feedback) is given by
\ba
\label{eq.thm.Cc.1}
C_{NFB} = \sup_{\bX} \uuline{I}(\bX;\bY)
\ea
where the supremum is over all sequences of finite-dimensional input distributions and $\uuline{I}(\bX;\bY)$ is the compound inf-information rate,
\ba
\label{eq.uuline{I}}
\uuline{I}(\bX;\bY) = \sup_{R}\bLF R: \lim_{n\rightarrow\infty} \sup_{s\in \sS} \Pr\left\{Z_{ns} \le R \right\} =0 \bRF
\ea
where $Z_{ns} = n^{-1}i(X^n;Y^n|s)$ is the normalized information density under channel state $s$.\qed
\end{thm}

This theorem was proved in \cite{Loyka-15}\cite{Loyka-15-2} using the Verdu-Han and Feinstein Lemmas properly extended to the compound channel setting.

For future use, we need the following formal definitions, which extend the respective $\inf$ and $\sup$ operators introduced for regular (single-state) sequences \cite{Verdu}\cite{Han} (see \eqref{eq.oln-oln-Xs}) to the compound setting.
\begin{defn}
\label{defn.ooline.uuline}
Let $\{X_{sn}\}_{n=1}^{\infty}$ be an arbitrary compound random sequence where $s$ is a state (i.e. a random sequence indexed by the state $s$). The compound infimum $\uuline{\{\cdot\}}$ and supremum $\ooline{\{\cdot\}}$ operators are defined as follows:
\ba
\label{eq.comp.inf-sup}
\uuline{\bX}=\uuline{\{X_{sn}\}} &= \sup \bLF x: \lim_{n\rightarrow\infty} \sup_{s} \Pr\left\{ X_{sn} \le x \right\} =0 \bRF\\
\ooline{\bX}=\ooline{\{X_{sn}\}} &= \inf \bLF x: \lim_{n\rightarrow\infty} \sup_{s} \Pr\left\{ X_{sn} \ge x \right\} =0 \bRF
\ea
\end{defn}
Roughly, $\uuline{\bX}$ and $\ooline{\bX}$ represent the largest lower and least upper bounds to the asymptotic support set of $X_{sn}$ over the whole state set (note $\sup_s$ in the definitions).

The following definitions extend the respective information-theoretic quantities in \cite{Verdu}\cite{Han} to the compound setting.

\begin{defn}
Let $\bX=\{X_s^n\}_{n=1}^{\infty}$ and $\bY=\{Y_s^n\}_{n=1}^{\infty}$ be two compound random sequences with distributions $p_{s x^n}$ and $p_{s y^n}$ where $s$ is a state.
The compound inf-divergence rate is defined as
\ba
\label{eq.D(x||y)}
\uuline{D}&(\bX;\bY) = \uuline{\bLF d_{sn}(X_s^n||Y_s^n) \bRF }
\ea
where $d_{sn}(x^n||y^n) =\frac{1}{n} \log \frac{p_{s x^n}(x^n)}{p_{s y^n}(x^n)}$ is the divergence density rate.
The compound inf and sup-entropy rates $\uuline{H}(\bX)$ and $\ooline{H}(\bX)$ are defined as
\ba
\label{eq.H(x)}
\uuline{H}(\bX) = \uuline{\{h_{sn}(X_s^n)\}},\ \ooline{H}(\bX) = \ooline{\{h_{sn}(X_s^n)\}}
\ea
where $h_{sn}(x^n) = - n^{-1} \log p_{s}(x^n)$ is the entropy density rate.
The compound conditional inf-entropy rate $\uuline{H}(\bY|\bX)$ and sup-entropy rate $\ooline{H}(\bY|\bX)$ are defined analogously via the conditional entropy density rate $h_{sn}(y^n|x^n) = - n^{-1} \log p_{s}(y^n|x^n)$ (with respect to the joint distribution $p_{s}(x^n, y^n)$), and $\ooline{I}(\bX;\bY)$ is similarly defined.
\end{defn}

The proposition below gives the properties of the compound inf-information rate $\uuline{I}(\bX;\bY)$ and other relevant quantities \cite{Loyka-15}\cite{Loyka-15-2}, which will be instrumental below.

\begin{prop}
\label{prop.properties.I}
Let $\bX$ and $\bY$ be (arbitrary) compound random sequences. The following holds:
\ba
\label{eq.properties.1}
&\uuline{D}(\bX||\bY) \ge 0 \\
\label{eq.properties.2}
&\ooline{I}(\bX;\bY) \ge \uuline{I}(\bX;\bY) \ge 0 \\
\label{eq.properties.3}
&\uuline{I}(\bX;\bY) = \uuline{I}(\bY;\bX) \\
\label{eq.properties.4}
&\uuline{I}(\bX;\bY) \le \ooline{H}(\bY) - \ooline{H}(\bY|\bX)\\
\label{eq.properties.5}
&\uuline{I}(\bX;\bY) \le \uuline{H}(\bY) - \uuline{H}(\bY|\bX)\\
\label{eq.properties.6}
&\uuline{I}(\bX;\bY) \ge \uuline{H}(\bY) - \ooline{H}(\bY|\bX)\\
\label{eq.properties.4a}
&\ooline{H}(\bY) \ge \ooline{H}(\bY|\bX)\\
\label{eq.properties.5a}
&\ooline{H}(\bY) \ge \uuline{H}(\bY) \ge \uuline{H}(\bY|\bX)
\ea
If the alphabets are discrete, then
\ba
\label{eq.properties.8}
&0\le \uuline{H}(\bX|\bY) \le \uuline{H}(\bX) \le \ooline{H}(\bX) \le \log N_x\\
\label{eq.properties.9}\notag
&0\le \uuline{I}(\bX;\bY) \le \min\{\uuline{H}(\bX), \uuline{H}(\bY)\}\\
 &\qquad\qquad\qquad \le \min\{\log N_x, \log N_y\} \\
\label{eq.properties.9a}\notag
&\uuline{I}(\bX;\bY) = \min\{\uuline{H}(\bX), \uuline{H}(\bY)\}\\
 &\qquad\qquad \mbox{if}\ \min\{\ooline{H}(\bY|\bX), \ooline{H}(\bX|\bY)\}=0 \\
\label{eq.properties.10}\notag
&0\le \ooline{I}(\bX;\bY) \le \min\{\ooline{H}(\bX), \ooline{H}(\bY)\}\\
 &\qquad\qquad\qquad \le \min\{\log N_x, \log N_y\}
\ea
where the last inequalities in \eqref{eq.properties.8}-\eqref{eq.properties.10} hold if the alphabets are of finite cardinality $N_x, N_y$.
\end{prop}

Note that many of these properties mimic the respective properties of mutual information and entropy, e.g. "conditioning cannot increase the entropy" and "mutual information is non-negative, symmetric and bounded by the entropy of the alphabet".

\section{Capacity with Feedback}

In this section, we consider a discrete compound channel with feedback and general additive noise. Instead of dealing with the feedback channel $\{w^n y^{k-1}\} \rightarrow x^k \rightarrow y_k$, $k=1...n$, directly, one can consider an effective channel $w^n \rightarrow y^n$ without feedback, see Fig. \ref{fig.1}. Applying Theorem \ref{thm.C.general} to the effective channel, the capacity with the feedback can be expressed as\footnote{see also \cite{Tatikonda-09} for a formulation based on the directed information for the case of full CSI and a proof of equivalence of these two formulations in the latter case.}
\ba
\label{eq.C_FB}
C_{FB} = \sup_{\bW,\bF} \uuline{I}(\bW;\bY)
\ea
where $\bY=\{Y^n\}_{n=1}^{\infty}$ and likewise for $\bW$, $\uuline{I}(\bW;\bY)$ is the compound inf-information rate:
\ba
\label{eq.uuline{I}}
\uuline{I}(\bW;\bY) = \uuline{\bLF n^{-1} i(W^n;Y^n|s)\bRF}
\ea
where $i(W^n;Y^n|s)$ is the information density:
\ba
i(w^n;y^n|s) = \log \frac{p_s(y^n|w^n)}{p_s(y^n)}
\ea
The maximization in \eqref{eq.C_FB} is over all possible encoding functions $\bF=\{f^n\}_{n=1}^{\infty}$ and all possible message distributions. Unfortunately, this maximization is difficult to perform in general. Therefore, we proceed in a different way. Let
\ba
\label{eq.ooline{H}}
\ooline{H}(\bXi) = \ooline{\bLF n^{-1} h(\Xi^n_s|s)\bRF}
\ea
be the compound sup-entropy rate of the compound noise $\bXi = \{\Xi^n_s\}_{n=1}^{\infty}$, $\Xi^n_s = \{\Xi_{1s}...\Xi_{ns}\}$, $h(\xi^n|s) = - \log p_s(\xi^n)$. The following is the main result of the paper.

\begin{thm}
\label{thm.C_FB=C_NFB}
The capacity of the compound discrete channel with (arbitrary) additive noise in \eqref{eq.ch.model} and the full Rx CSI is not increased by the causal feedback:
\ba
\label{eq.C_FB=C_NFB}
C_{FB} = C_{NFB} = \sup_{\bX} \uuline{I}(\bX;\bY) = \log M - \ooline{H}(\bXi)
\ea
where $C_{NFB}$ is the capacity without feedback, and $\sup_{\bX}$ is over all sequences of input distributions.
\end{thm}
\begin{proof}
Let us consider the no-feedback case first. The 2nd equality follow from Theorem 1. The following Lemma is needed to prove the last equality.
\begin{lemma}
\label{lemma.z_k}
Let $z_k = g_{sk}(x^k)$, $k=1...n$, and the mapping $x^n \rightarrow z^n$ is one-to-one. If $p(x^n)=1/M^n$, then $p(z^n)=1/M^n$, i.e. equiprobable $X^n$ generates equiprobable $Z^n$.
\end{lemma}
Now, since the mapping $x^n \rightarrow z^n$ is invertible, $\uuline{I}(\bX;\bY) = \uuline{I}(\bZ;\bY)$. It follows from \eqref{eq.properties.4} that
\ba
\uuline{I}(\bZ;\bY) \le \ooline{H}(\bY) - \ooline{H}(\bY|\bZ)
\ea
Using this inequality, one obtains:
\ba
\uuline{I}(\bZ;\bY) &\le \log M - \ooline{H}(\bY|\bZ)\\
&= \log M - \ooline{H}(\bXi)
\ea
where 1st inequality is due to $M$-ary alphabets, so that $\ooline{H}(\bY) \le \log M$ (see \eqref{eq.properties.8}), and the equality is due to $\ooline{H}(\bY|\bZ) = \ooline{H}(\bZ + \bXi|\bZ) = \ooline{H}(\bXi)$, since the noise is additive and independent of the input (recall that we consider the no feedback case). Finally,
\ba
\uuline{I}(\bX;\bY) \le \log M - \ooline{H}(\bXi)
\ea
and the equality is achieved by equiprobable input due to Lemma  \ref{lemma.z_k}, under which the output is also equiprobable. This proves the last equality in \eqref{eq.C_FB=C_NFB}.

To prove 1st equality, $C_{FB}=C_{NFB}$, observe that feedback cannot decrease the capacity,
\ba
\log M - \ooline{H}(\bXi) = C_{NFB} \le C_{FB}
\ea
To prove the converse,
\ba
\label{eq.C_FB<}
C_{FB} \le \log M - \ooline{H}(\bXi)
\ea
use \eqref{eq.C_FB} to conclude
\ba
C_{FB} &= \sup_{\bW,\bF} \uuline{I}(\bW;\bY)\\
&\le \sup_{\bW,\bF} [\ooline{H}(\bY) - \ooline{H}(\bY|\bW)]\\
\label{eq.C_FB<2}
&\le \log M - \inf_{\bW,\bF}\ooline{H}(\bY|\bW)
\ea
where 1st inequality is due to $\uuline{I}(\bW;\bY) \le \ooline{H}(\bY) - \ooline{H}(\bY|\bW)$  and 2nd inequality is due to $\ooline{H}(\bY) \le \log M$ (since the alphabet is $M$-ary).

To evaluate $\ooline{H}(\bY|\bW)$, note that
\ba
p_s(y^n|w^n) = \prod_{k=1}^{n} p_s(y_k|y^{k-1}w^n)
\ea
and
\ba
p_{s}(y_k|y^{k-1}w^n) &= p_{s}(y_k|y^{k-1} x^k w^n)\\
&= p_{s}(y_k|y^{k-1} x^k \xi^{k-1} w^n)\\
&= p_{sy}\left(g_{sk}(x^k) +\xi_k|x^k \xi^{k-1} w^n\right)\\
&= p_{s\xi}(\xi_k|\xi^{k-1} w^n)\\
&= p_{s\xi}(\xi_k|\xi^{k-1})
\ea
where $\xi_k = y_k - g_{sk}(x^k)$, $x^k=f^k(w^n y^{k-1})$, $\xi^n = \{\xi_k\}_{k=1}^n$. 1st equality is due to $x^k = f^k( w^n y^{k-1})$; 2nd and 3rd equalities are due to the channel model $y_k= z_k + \xi_k$; 4th equality is due to $x^k = \check{f}^k(w^n \xi^{k-1})$, where $\check{f}^k$ is a function which depends on encoding functions $f^k$ and the channel impulse response functions $g_s^k$; last equality is due to independence of noise and message. Thus,
\ba
\label{eq.p(sy)=p(xi)}
p_{sy}(y^n|w^n) = p_{s\xi}(\xi^n)
\ea
and therefore
\ba
\ooline{H}(\bY|\bW) = \ooline{H}(\bXi)
\ea
Combining this with \eqref{eq.C_FB<2}, one obtains \eqref{eq.C_FB<} and hence the desired result follows.

Equality in \eqref{eq.C_FB<} is achieved by the uniform input distribution $p(x^n)=1/M^n$, which is also i.i.d. and equiprobable: $p(x^n)=\prod_{i=1}^n p(x_i)$, $p(x_i)=1/M$ (this can be shown by induction), under which the output is also uniform.
\end{proof}

\begin{remark}
Note that in both feedback and no-feedback systems, the optimizing input is uniform and hence i.i.d. equiprobable and independent of the feedback. Under this input, the output is also uniform under \textit{any} noise, which explains why feedback is not helpful in this setting.
\end{remark}

\begin{remark}
Setting $l_s=0$ in \eqref{eq.ch.model.z}, one obtains a channel without intersymbol interference. When, in addition, the uncertainty set $\mathcal{S}$ is singleton (single-state channel with no uncertainty), Theorem \ref{thm.C_FB=C_NFB} above reduces to the corresponding result in \cite{Alajaji-94} obtained for fully known (no uncertainty) channels.
\end{remark}

\begin{remark}
Since noisy feedback cannot perform better than noiseless, this result also implies that noisy feedback cannot increase the compound capacity in this setting either.
\end{remark}

\begin{remark}
One may consider a more general feedback of the form $u_k=\beta_k(y^k)$, where $\{\beta_k\}$ are arbitrary (possibly random) feedback functions (which account for e.g. quantization of feedback signals and noise in the feedback channel), and the corresponding encoding of the form $x_k = f_k(w^n u^{k-1})$. Since the capacity with this form of feedback cannot exceed the capacity with the full feedback of $y^{k-1}$, Theorem \ref{thm.C_FB=C_NFB} still holds for this setting as well.
\end{remark}

\section{Impact of the Tx CSI and Saddle Point}

Let us consider the case where channel state $s$ is known at the transmitter, so that codewords can be selected as functions of the channel state. In this case, the worst-case channel capacity $C_w$ is a proper performance metric and it can be expressed as
\ba
\label{eq.Cw.1}
C_w &= \inf_s \sup_{\bW, \bF} \underline{I}(\bW;\bY|s)\\
\label{eq.Cw.2}
&= \inf_s (\log M - \overline{H}(\bXi|s))\\
\label{eq.Cw.2a}
&= \inf_s \sup_{\bX} \underline{I}(\bX;\bY|s)\\
\label{eq.Cw.3}
&= \log M - \sup_s \overline{H}(\bXi|s)\\
\label{eq.Cw.4}
&\ge \log M - \ooline{H}(\bXi) = C_{FB}
\ea
where $\underline{I}(\bX;\bY|s)$ is the inf-information rate under channel state $s$ \cite{Verdu}:
\ba
\label{eq.uline{I}}
\underline{I}(\bX;\bY|s) = \sup_{R}\bLF R: \lim_{n\rightarrow\infty}  \Pr\left\{\frac{1}{n} i(W^n;Y^n|s) \le R \right\} =0 \bRF,
\ea
$\overline{H}(\bXi|s)$ is the sup-entropy rate of the noise under state $s$:
\ba
\label{eq.ooline{H}}
\overline{H}(\bXi|s) = \inf_{R}\bLF R: \lim_{n\rightarrow\infty} \Pr\left\{\frac{1}{n} h(\Xi^n_s|s) > R \right\} =0 \bRF
\ea
\eqref{eq.Cw.1} follows from the general formula in \cite{Verdu} and the equivalent channel in Fig. \ref{fig.1}; \eqref{eq.Cw.2} follows from the Theorem in \cite{Alajaji-94}; \eqref{eq.Cw.4} follows from the Lemma \ref{lemma.ooH-supHs} below, so that the impact of Tx CSI can be characterized by
\ba
\label{eq.DeltaC}
\Delta C = C_w - C_{FB} = \ooline{H}(\bXi) - \sup_s \overline{H}(\bXi|s) \ge 0
\ea

Note that, similarly to the compound capacity, $C_w$ is not increased by the feedback either, i.e. \eqref{eq.Cw.3} is also the no-feedback worst-case channel capacity as indicated by \eqref{eq.Cw.2a} while $C_s= \sup_{\bX} \underline{I}(\bX;\bY|s)$ is the channel capacity under state $s$ known to both Tx and Rx.

To proceed further, we need the following definition.

\begin{defn}
The compound noise sequence $\{\Xi_s^n\}_{n=1}^{\infty}$ is uniform if the convergence in
\ba
\label{eq.def.unif}
\Pr\left\{\frac{1}{n} h(\Xi^n_s|s) > \sup_s \overline{H}(\bXi|s) +\delta \right\} \rightarrow 0
\ea
as $n \rightarrow\infty$ is uniform in $s \in \sS$ for any $\delta>0$.
\end{defn}

Note that, while the convergence to zero in \eqref{eq.def.unif} for each $\delta>0$ and $s \in \sS$ is guaranteed from the definition of $\sup_s \overline{H}(\bXi|s)$, this convergence does not have to be uniform in general. In fact, the uniform convergence requirement above is equivalent to
\ba
\label{eq.def.unif-2}
\lim_{n\rightarrow\infty} \sup_s \Pr\left\{\frac{1}{n} h(\Xi^n_s|s) > \sup_s \overline{H}(\bXi|s) +\delta \right\} = 0
\ea
for any $\delta>0$, which is clearly stronger than just point-wise convergence in \eqref{eq.def.unif} for each $s$, which is equivalent to
\ba
\label{eq.def.unif-3}
\sup_s  \lim_{n\rightarrow\infty} \Pr\left\{\frac{1}{n} h(\Xi^n_s|s) > \sup_s \overline{H}(\bXi|s) +\delta \right\} = 0
\ea
In general, $\lim$ and $\sup$ cannot be swaped; rather
\ba
\label{eq.suplim<=limsup}
\sup_s  \lim_{n\rightarrow\infty} \{\cdot\} \le \lim_{n\rightarrow\infty} \sup_s \{\cdot\}
\ea
so that \eqref{eq.def.unif-2} implies \eqref{eq.def.unif-3} but the converse is not true in general, i.e. the inequality can be strict.

We are now in a position to establish the following key result.

\begin{lemma}
\label{lemma.ooH-supHs}
The following inequality holds for the general compound noise sequence:
\ba
\label{eq.ooH-supHs}
\ooline{H}(\bXi) \ge \sup_s \overline{H}(\bXi|s)
\ea
with equality if and only if the compound noise is uniform.
\end{lemma}
\begin{proof}
See Appendix.
\end{proof}

Strict inequality in \eqref{eq.ooH-supHs} can be demonstrated by examples - see Section \ref{sec.Examples}. Combining Lemma \ref{lemma.ooH-supHs} with \eqref{eq.DeltaC}, one obtains the following result.

\begin{thm}
Consider the discrete compound channel with additive noise as in \eqref{eq.ch.model} under the full Rx CSI. When the compound noise is uniform, neither the full Tx CSI nor causal noiseless or noisy feedback increase its capacity, i.e.
\ba
C_{NFB} = C_{FB} = C_w
\ea
\end{thm}

The last equality states that the worst-case channel capacity (achievable by codebooks tailored to the channel state) is the same as the compound channel capacity (where the codebooks are independent of channel states), which can be equivalently expressed as
\ba
\inf_s \sup_{\bX} \underline{I}(\bX;\bY|s) = \sup_{\bX} \inf_s \underline{I}(\bX;\bY|s)
\ea
so that, when $\inf$ and $\sup$ are achieved, this is equivalent to the existence of a saddle point \cite{Boyd-04}\cite{Zeidler-86}:
\ba
\underline{I}(\bX;\bY|s^*)\le \underline{I}(\bX^*;\bY^*|s^*)\le \underline{I}(\bX^*;\bY^*|s)
\ea
where $(\bX^*,s^*)$ is the saddle point. This saddle point exists for both feedback and no feedback cases when the compound noise is uniform.

It is remarkable that a saddle point exists even though the uncertainty set is allowed to be non-convex and the objective function $f(s)\triangleq \underline{I}(\bX;\bY|s)$ is not required to be convex either (e.g. when $s$ is discrete, $f(s)$ is not convex; it can also be non-convex even when the uncertainty set is convex), so that von Neumann's mini-max Theorem \cite{Boyd-04} or its extensions \cite{Zeidler-86} can not be used to establish the existence of a saddle point.

The saddle point above extends the information-theoretic saddle-point results established earlier in e.g. \cite{Dobrushin-59-2}-\cite{Loyka-15-3} for stationary and ergodic (and hence information-stable) channels, for which mutual information is a proper metric, to the realm of information-unstable scenarios, where mutual information has no operational meaning and the inf-information rate has to be used instead. Furthermore, it demonstrates that neither convexity of the feasible set nor of the objective function are necessary for the saddlepoint to exist. It also has the standard game-theoretic interpretation: neither the nature (who controls state $s$) nor the transmitter (who controls the input distribution) can deviate from the optimal strategy without incurring a penalty.

\section{Examples}
\label{sec.Examples}

In this section, we consider some illustrative examples. Among other things, they identify the scenarios when the Tx CSI increases the capacity and when it does not.

\subsection{Example 1}
Let the compound noise be of the form
\ba
\label{eq.xi.E1}
\xi_{s}^n = \{w_1,..,w_s,0..0\}
\ea
where $W_i$ are i.i.d. equiprobable so that $p(w^s)=1/M^s$, and $s \in \{1,2,...\}$. This can model block interference/noise of length $s$. Note that the noise process $\{\xi_{s}^n\}$ is not stationary. Using \eqref{eq.ooline{H}}, one obtains, after some manipulations, $\overline{H}(\bXi|s)= 0 \ \forall s$ (this is due to the fact that, under fixed $s$, the "noisy" part in \eqref{eq.xi.E1} is asymptotically negligible) so that $\sup_s \overline{H}(\bXi|s)= 0$ and hence
\ba
C_w = \log M
\ea
Yet, using \eqref{eq.ooline{H}}, it follows, after some manipulations, that $\ooline{H}(\bXi) = \log M> \sup_s \overline{H}(\bXi|s)= 0$. Hence, the noise is not uniform and
\ba
C_{FB} = 0
\ea
so that the advantage of the Tx CSI is significant: $\Delta C = \log M$, i.e. the maximum possible value for $M$-ary alphabet. The reason for this is that the compound noise in \eqref{eq.xi.E1} is not uniform, the worst-case noise (corresponding to $\sup_s$ in \eqref{eq.comp.inf-sup}) is i.i.d. equiprobable for any given $n$ and hence the compound channel is useless, even under noiseless causal feedback. The presence of the Tx CSI changes the situation dramatically: one can now design a codebook for any given state $s$ and make the error probability arbitrary low by using sufficiently large blocklength $n \gg s$. This conclusion also holds for \textit{any} distribution of $w_s^n$, not only i.i.d. equiprobable, since $\sup_s \overline{H}(\bXi|s)= 0$ regardless.

The situation also changes dramatically if one imposes the boundedness constraint on the uncertainty set: $s \le S < \infty$. In this case, $\ooline{H}(\bXi) = \sup_s \overline{H}(\bXi|s)= 0$, i.e. the noise becomes uniform, and hence $C_{FB}=C_w= \log M$. One may wonder as to what are the practical implications of these dramatic changes. In our view, the first case of unbounded $s$ corresponds to a scenario where the interference is more powerful than the codebook, i.e. for any given $n$, does not matter how large, one can always find powerful enough interference with $s=n$ thus rendering the channel useless. The second case of bounded $s$ prevents this thus allowing for the codeword length $n$ to be much larger than $S$ and hence represents a scenario where the codebook is more powerful than any possible interference. In the same way, one can interpret the impact of the Tx CSI: giving $s$ to the Tx allows one to chose $n \gg s$ and hence make the impact of interference negligible, which is not possible otherwise.

\subsection{Example 2}
Let us now set the compound noise as
\ba
\label{eq.xi.E2}
\xi_{s}^n = \{w_1,..,w_s,z_{s+1},..,z_n\}
\ea
with binary alphabet and $W_i \sim \textsf{Ber}(p_1)$, $Z_i \sim \textsf{Ber}(p_2)$, i.e. Bernoulli random variables, all independent of each other and $0\le p_2 < p_1 \le 1/2$, and $s \in \{1,2,...\}$. This can model a scenario where there is noise (2nd part) in addition to interference (1st part).

One obtains, after some manipulations,
\ba
\sup_s \overline{H}(\bXi|s)= h(p_2) < h(p_1)=\ooline{H}(\bXi)
\ea
where $h(p)$ is the binary entropy function, and hence
\ba
\Delta C = C_w - C_{FB} = h(p_1) - h(p_2) >0
\ea
so that the noise is not uniform and the Tx CSI does bring in advantage. Bounding the uncertainty set $s\le S< \infty$ has no impact on $\overline{H}(\bXi|s)$ but makes $\ooline{H}(\bXi) = h(p_2)$ and hence the advantage of the Tx CSI disappear: $\Delta C=0$. The noise becomes uniform in this case. As in Example 1, the distribution of $w^s$ does not affect $\overline{H}(\bXi|s)$ but does have an impact on $\ooline{H}(\bXi)$.

If, on the other hand, $p_2 \ge p_1$, then $\Delta C = 0$ regardless whether the uncertainty set is bounded or not, so that, in general,
\ba
\Delta C = C_w - C_{FB} = [h(p_1) - h(p_2)]_+
\ea
where $[x]_+=\max\{0,x\}$.

\subsection{Example 3}
Let the compound noise sequence be of the form
\ba
\label{eq.xi.E3}
\xi_{s}^n = \{w_1,..,w_n\}
\ea
with binary alphabet and $W_i \sim \textsf{Ber}(p_i)$ and independent of each other, where
\ba
\label{eq.pi}
p_i = \frac{s}{2(i+s)}
\ea
and $s\ge 0$ (not necessarily integer). This models a scenario where noise becomes "weaker" with time (note that $h(p_i)$ decreases with $i$), while $s$ controls the decay rate: noise becomes negligible when $i\gg s$, so that $h(p_i) \approx 0$, see Fig. \ref{fig.hpi}. The process is clearly not stationary.

\begin{figure}[t]
\centerline{\includegraphics[width=4in]{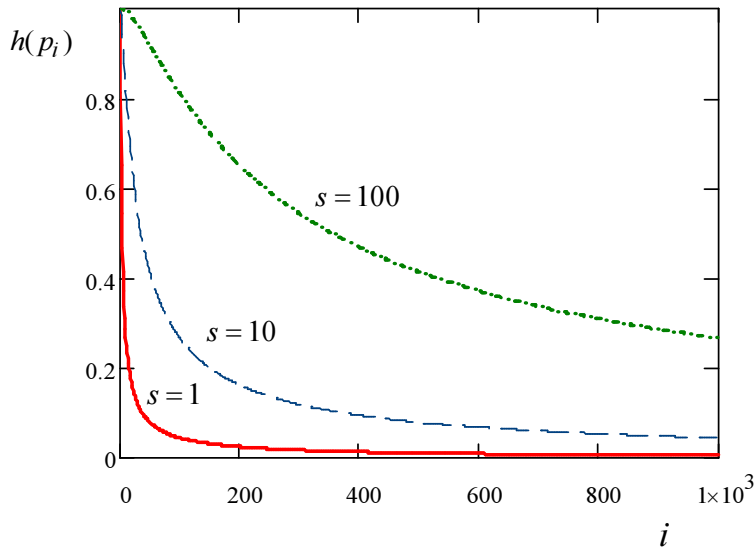}}
\caption{Binary entropy $h(p_i)$ for $p_i$ as in \eqref{eq.pi} versus time $i$ for different states $s$.}
\label{fig.hpi}
\end{figure}

After some manipulations, one obtains $\overline{H}(\bXi|s) = 0\ \forall\ s$ and hence $\sup_s\overline{H}(\bXi|s)=0$. Yet, $\ooline{H}(\bXi)=1$ so that $C_w=1$, $C_{FB}=0$  and $\Delta C= 1$, the maximal possible value, and the noise in not uniform. Thus, while the Tx CSI is the most useful, the noiseless causal feedback is useless.

As above, bounding the uncertainty set $s\le S< \infty$ changes the situation dramatically: $C_{FB}=1$, $\Delta C =0$, so that the Tx CSI gives no increase in the capacity since the noise is now uniform.

\subsection{Example 4}
Let us now consider a non-ergodic non-stationary channel with
\ba
\label{eq.xi.E4}
\xi_{s}^n = \begin{cases}
                w_s^n & \text{with}\ \Pr = p\\
                z_s^n & \text{with}\ \Pr = 1-p
            \end{cases}
\ea
where $W_s^n, Z_s^n$ are non-stationary processes (e.g. from the above examples) and $0<p<1$, i.e. one of the two processes is randomly selected at the beginning and it operates during the entire transmission. This process is clearly non-ergodic when its components are of different distributions. One may also consider $p_s$, i.e. a function of the channel state, provided that $0<\alpha \le p_s \le \beta <1\ \forall s$.

It can be seen that
\ba
\overline{H}(\bXi|s)= \max\{\overline{H}(\bW|s), \overline{H}(\bZ|s)\},\ \ooline{H}(\bXi)= \max\{\ooline{H}(\bW), \ooline{H}(\bZ)\}
\ea
so that
\ba
\Delta C = \max\{\ooline{H}(\bW), \ooline{H}(\bZ)\}- \sup_s \max\{\overline{H}(\bW|s), \overline{H}(\bZ|s)\}
\ea
In particular, $\Delta C=0$ if $\{W_s^n\}_{n=1}^{\infty}, \{Z_s^n\}_{n=1}^{\infty}$ are uniform compound sequences. This holds if e.g. the uncertainty set is of a finite cardinality, regardless of what the distributions of $\{W_s^n\}, \{Z_s^n\}$ are.

\section{Strong Converse}

In this section, we establish a sufficient and necessary condition for the strong converse to hold for the compound channel with additive noise. In addition to being of theoretical interest on its own, it also has some practical implications. In particular, strong converse ensures that slightly larger error probability cannot be traded off for higher data rate, since the transition from arbitrary low to high error probability is sharp. Additionally, a consequence of this is that the error rate performance degrades dramatically if the SNR drops below the threshold for which the system was designed.

Let $\varepsilon_n$ and $r_n$ be the error probability and rate of a codebook of blocklength $n$. The formal definition of strong converse is as follows.

\begin{defn}
A compound channel is said to satisfy strong converse if
\ba
\label{eq.strong.conv.d1}
\lim_{n\rightarrow\infty} \varepsilon_n = 1
\ea
for any code satisfying
\ba
\label{eq.strong.conv.d2}
\liminf_{n\rightarrow\infty} r_n > C_c
\ea
\end{defn}

We begin with the following definitions which are needed below. 1st one extends the standard definition of convergence in probability to compound random sequences.

\begin{defn}
A compound random sequence $\{Y_{sn}\}_{n=1}^{\infty}$ is said to converge in probability to $y_0$, denoted as $Y_{sn} \overset{\Pr} {\rightarrow} y_0$, if
\ba
\label{eq.uni.conv.def}
\lim_{n\rightarrow\infty}\sup_s \Pr\{|Y_{sn}-y_0|>\epsilon\}=0
\ea
for any $\epsilon>0$, where $\sup_s$ is over the whole state set.
\end{defn}

It should be emphasized that the point-wise convergence, i.e. $\lim_{n\rightarrow\infty} \Pr\{|Y_{sn}-y_0|>\epsilon\}=0\ \forall s$, does not imply \eqref{eq.uni.conv.def}, which is a stronger condition (see also \eqref{eq.suplim<=limsup}).

In addition to the following standard definitions of the infimum $\uln{\bX}_s$ and supremum $\oln{\bX}_s$ of a random sequence $X_s^n$ under state $s$ \cite{Verdu}\cite{Han}:
\ba\notag
\label{eq.oln-oln-Xs}
\uln{\bX}_s &= \sup \bLF x: \lim_{n\rightarrow\infty} \Pr\left\{X_{sn} \le x \right\} =0 \bRF\\
\oln{\bX}_s &= \inf \bLF x: \lim_{n\rightarrow\infty} \Pr\left\{X_{sn} \ge x \right\} =0 \bRF
\ea
and the compound infimum $\uuline{\bX}$ and supremum $\ooline{\bX}$ in Definition \ref{defn.ooline.uuline}, the following compound $\inf$ and $\sup$ operators are needed in a condition for strong converse.
\begin{defn}
\label{defn.cis.oper}
Let $\{X_{sn}\}_{n=1}^{\infty}$ be an arbitrary compound random sequence where $s$ is a state. The compound infimum $\uln{\{\cdot\}}$ and supremum $\oln{\{\cdot\}}$ operators are defined as follows:
\ba
\label{eq.comp.inf-sup.BW}
\uln{\bX}=\uln{\{X_{sn}\}} &= \sup \bLF x: \lim_{n\rightarrow\infty} \inf_{s} \Pr\left\{ X_{sn} \le x \right\} =0 \bRF\\
\oln{\bX}=\oln{\{X_{sn}\}} &= \inf \bLF x: \lim_{n\rightarrow\infty} \inf_{s} \Pr\left\{ X_{sn} \ge x \right\} =0 \bRF
\ea
\end{defn}
Roughly, $\uln{\bX}_s$ and $\oln{\bX}_s$  represent the largest lower and least upper bounds of the asymptotic support set of $X_{sn}$ under state $s$ while $\uln{\bX}$ and $\oln{\bX}$ do so over the whole state set by selecting the best states for the respective bounds. Note however that these quantities are different from $\uuline{\bX}$ and $\ooline{\bX}$:  $\inf$ rather than $\sup$ are used in the definitions of $\uln{\bX}$ and $\oln{\bX}$ so that the respective limits are enforced for some channel states only, not over the whole state set. While subtle, the difference is important, as we will see below. These operators have the properties which are instrumental in establishing the strong converse and other results.

\begin{prop}
\label{prop.sc.prop}
The compound $\inf$ and $\sup$ operators in Definition \ref{defn.cis.oper} satisfy the following:
\ba
\label{eq.sc.P1}
&\uln{(-\bX)} = -\oln{\bX}\\
\label{eq.sc.P2}
&\oln{\bX}+\uuline{\bY} \le \oln{(\bX+\bY)} \le \oln{\bX}+\ooline{\bY}\\
\label{eq.sc.P3}
&\uuline{\bX} \le \min\{\uln{\bX},\oln{\bX}\} \le \max\{\uln{\bX},\oln{\bX}\} \le \ooline{\bX}\\
\label{eq.sc.P4}
&\sup_s \underline{\bX}_s \le \uln{\bX},\ \oln{\bX} \le \inf_s \overline{\bX}_s
\ea
If $Y_{sn} \overset{\Pr} {\rightarrow} y_0$, then
\ba
\label{eq.sc.P5}
\oln{(\bX+\bY)} = \oln{\bX}+y_0
\ea
\end{prop}
\begin{proof}
See Appendix.
\end{proof}

Strict inequalities in Proposition \ref{prop.sc.prop} can be demonstrated via examples.

\textbf{Example 1}: Let $X_{1n}$ and $X_{2n}$ be uniformly-distributed random variables,
\ba
X_{1n} \sim  \textsf{uni}[0,2],\ X_{2n} \sim  \textsf{uni}[1,3]
\ea
so that
\ba
\uuline{\bX}=0,\ \uln{\bX}=1,\ \oln{\bX}=2,\ \ooline{\bX}=3
\ea
and all inequalities in \eqref{eq.sc.P3} are strict. Since
\ba
\sup_s \uln{\bX}_s=1,\ \inf_s \oln{\bX}_s=2,
\ea
this example also demonstrates that the inequalities in \eqref{eq.sc.P4} can become equalities.

To demonstrate that the inequalities in \eqref{eq.sc.P2} can be strict, set $X_{sn}$ to be deterministic constant and $Y_{sn}$ as $X$ in Example 1 above.

\textbf{Example 2}: To see that the inequalities in \eqref{eq.sc.P4} can be strict, let $X_{sn}$ be Bernoulli random variables as follows:
\ba
X_{sn} \sim  \textsf{Ber}\{1-p_{sn}\},\ p_{sn} = \frac{n}{n+s},\ s \ge 0
\ea
so that $\Pr\{X_{sn}=0\}=p_{sn}$. It is straightforward to see that $\uln{\bX}_s=\oln{\bX}_s=0$ for any $s$ so that
\ba
\sup_s \uln{\bX}_s= \inf_s \oln{\bX}_s=0,
\ea
yet $\uln{\bX}=1$ so that 1st inequality is strict while 2nd one becomes equality since $\oln{\bX}=0$. To see that this inequality can be strict, set $X_{sn} \sim  \textsf{Ber}\{p_{sn}\}$ instead, so that
\ba
\sup_s \uln{\bX}_s= \inf_s \oln{\bX}_s=1,
\ea
yet $\oln{\bX}=0$.

Using Example 1 and its modifications, see Fig. \ref{fig.H-A} and \ref{fig.H-B}, one can also demonstrate that there is no specific relationship between $\uln{\bX}$ and $\oln{\bX}$ in general, i.e. neither $\uln{\bX} \le \oln{\bX}$ nor $\uln{\bX} \ge \oln{\bX}$ are true, unlike $\uuline{\bX} \le \ooline{\bX}$ that holds in full generality. In a similar way, it can be shown that there exists no specific relationship between $\sup_s \oln{\bX}_s$ and $\uln{\bX}$. This also holds for $\inf_s \uln{\bX}_s$ and $\oln{\bX}$.

\begin{figure}[t]
\centerline{\includegraphics[width=2.5in]{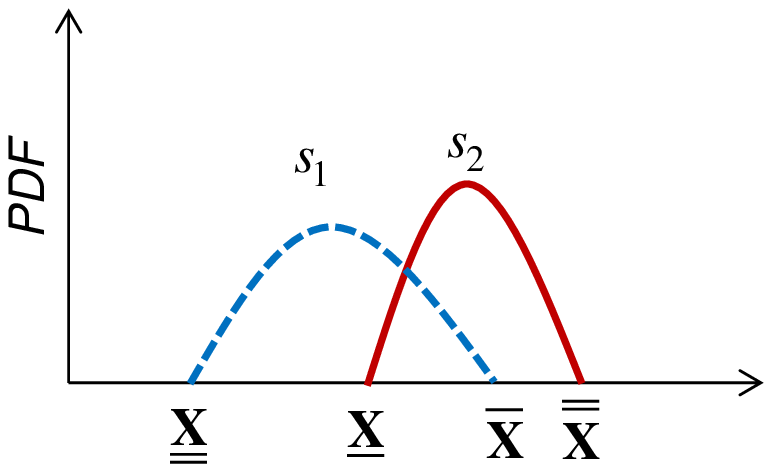}}
\caption{Asymptotic distribution of a random 2-state sequence $\bX$ and related quantities. Note that $\uln{\bX} < \oln{\bX}$.}
\label{fig.H-A}
\end{figure}

\begin{figure}[t]
\centerline{\includegraphics[width=2.5in]{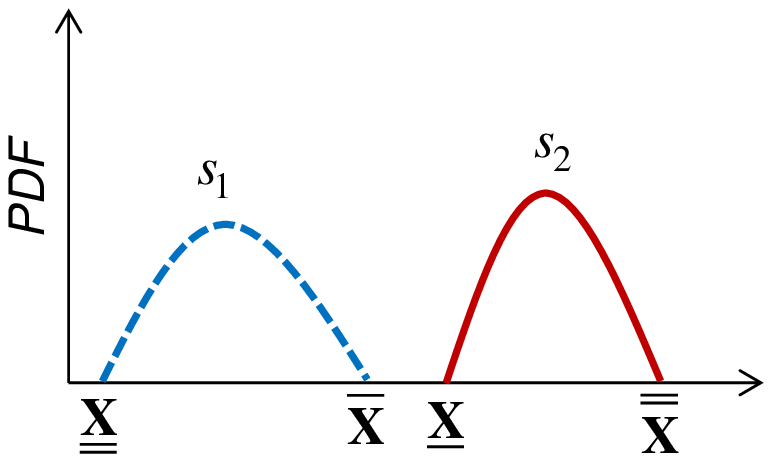}}
\caption{Asymptotic distribution of a random 2-state sequence $\bX$ and related quantities. Note that $\uln{\bX} > \oln{\bX}$.}
\label{fig.H-B}
\end{figure}

Using Proposition 9 in \cite{Loyka-15-2} and \eqref{eq.sc.P1}, the inequalities in \eqref{eq.sc.P4} can be refined as follows:
\ba\notag
\label{eq.sc.P4a}
&\uuline{\bX} \le \inf_s \uln{\bX}_s \le \sup_s \underline{\bX}_s \le \uln{\bX} \le \ooline{\bX},\\
&\uuline{\bX} \le \oln{\bX} \le \inf_s \overline{\bX}_s \le \sup_s \oln{\bX}_s \le \ooline{\bX}
\ea

A special case of \eqref{eq.sc.P2} is when $Y_{sn} = b$, i.e. a constant, so that, for any $a\ge 0$,
\ba
\oln{(a\bX+b)}=a\oln{\bX}+ b
\ea
i.e. $\oln{\{\cdot\}}$ is a linear operator for positive $a$. It is straightforward to see that, for negative $a$,
\ba
\oln{(a\bX+b)}=a\uln{\bX}+ b
\ea

Let $\uln{H}(\bXi) = \uln{\{n^{-1}h(\Xi_s^n|s)\}}$ and likewise for $\oln{H}(\bXi)$. In addition to its properties inherited from Proposition \ref{prop.sc.prop}, it also satisfies
\ba
0 \le \uln{H}(\bXi), \oln{H}(\bXi) \le \log M
\ea
where 1st inequality holds in full generality and 2nd one - for $M$-ary alphabets. We are now in a position to establish a sufficient and necessary condition for the strong converse to hold.

\begin{thm}
\label{thm.sc}
The compound channel with additive noise in \eqref{eq.ch.model} under the full Rx CSI satisfies the strong converse condition for both feedback and no feedback cases if and only if
\ba
\label{eq.thm.SC.1}
\ooline{H}(\bXi)=\uln{H}(\bXi)
\ea
If the compound noise is uniform, this reduces to
\ba
\sup_s\oln{H}(\bXi|s)=\uln{H}(\bXi)
\ea
\end{thm}
\begin{proof}
See Appendix.
\end{proof}

Fig. \ref{fig.SC} illustrates the condition of strong converse for a 2-state channel.

\begin{figure}[t]
\centerline{\includegraphics[width=2.5in]{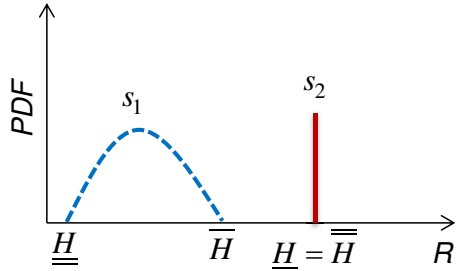}}
\caption{Asymptotic distribution of the noise entropy density rate for a 2-state channel with strong converse and related entropy density rates.}
\label{fig.SC}
\end{figure}

Using Proposition 27 in \cite{Loyka-15-2} under the optimal (uniform) input $\bX^*$ in combination with \eqref{eq.strong.conv.2}, one further obtains under the strong converse condition \eqref{eq.thm.SC.1}:
\ba
\label{eq.HbXi=H}
\ooline{H}(\bXi)= \limsup_{n\rightarrow\infty} \sup_{s} \frac{1}{n} H(\Xi_{s}^n)
\ea
where $H(\Xi_{s}^n)$ is the ergodic entropy, i.e. the compound sup-entropy rate $\ooline{H}(\bXi)$ coincides with the ergodic entropy rate of the noise (under its worst states), even though no ergodicity (or information stability) was imposed on the noise upfront\footnote{Note also that \eqref{eq.HbXi=H} equates two very different quantities: while the definition of $H(\Xi_{s}^n)$ is based on the expectation, so it is an ergodic quantity, that of $\ooline{H}(\bXi)$ does not use expectation at all.}. Hence, one concludes that the strong converse condition forces the worst-case noise to behave ergodically and hence the worst-case noise ergodicity is both necessary and sufficient for the strong converse to hold. This conclusion holds for both feedback and no feedback cases.

While there is no specific ordering between $\uln{H}(\bXi)$ and $\oln{H}(\bXi)$ or between $\sup_s \oln{H}(\bXi|s)$ and $\uln{H}(\bXi)$ in general (as indicated by the examples above), such ordering is induced by the strong converse, as indicated below.

\begin{cor}
Under the strong converse condition in Theorem \ref{thm.sc}, the following ordering holds:
\ba
\label{eq.cor.sc}
\oln{H}(\bXi) \le \inf_s \oln{H}(\bXi|s)\le \sup_s \oln{H}(\bXi|s) \le  \uln{H}(\bXi)
\ea
which is thus a necessary condition for the strong converse to hold.
\end{cor}
\begin{proof}
It follows from \eqref{eq.thm.SC.1} and \eqref{eq.sc.P4} that
\ba\notag
\uln{H}(\bXi) = \ooline{H}(\bXi) &\ge \sup_s \oln{H}(\bXi|s)\\
&\ge \inf_s \oln{H}(\bXi|s)\\ \notag
&\ge \oln{H}(\bXi)
\ea
\end{proof}

\subsection{Examples}

To gain further insight, one may use the examples of Section \ref{sec.Examples}. In particular, one obtains for Example 1
\ba
\ooline{H}(\bXi) = \uln{H}(\bXi) =\log M
\ea
when the uncertainty set is not bounded and
\ba
\ooline{H}(\bXi) = \uln{H}(\bXi) =0
\ea
when it is, so that the strong converse holds in both cases.

For Example 2,
\ba
\ooline{H}(\bXi) = \uln{H}(\bXi) = h(p_1)
\ea
when the uncertainty set is not bounded and
\ba
\ooline{H}(\bXi) = \uln{H}(\bXi) = h(p_2)
\ea
when it is, so that the strong converse holds in both cases as well.

For Example 3,
\ba
\ooline{H}(\bXi) = \uln{H}(\bXi) = 1
\ea
when the uncertainty set is not bounded and
\ba
\ooline{H}(\bXi) = \uln{H}(\bXi) = 0
\ea
when it is, so that the strong converse holds in both cases too.

Example 4 is more interesting. It is not too difficult to show that, in the general case,
\ba
\label{eq.ulnH(bXi)}
\uln{H}(\bXi) \le \min \{\uln{H}(\bW),\uln{H}(\bZ)\}
\ea
so that
\ba\notag
\uln{H}(\bXi) &\le \min \{\uln{H}(\bW),\uln{H}(\bZ)\}\\ \notag
 &\le \max \{\uln{H}(\bW),\uln{H}(\bZ)\}\\
  &\le \ooline{H}(\bXi) = \max\{\ooline{H}(\bW),\ooline{H}(\bZ)\}
\ea
and hence, if $\uln{H}(\bW) \neq \uln{H}(\bZ)$,
\ba
\uln{H}(\bXi) < \ooline{H}(\bXi)
\ea
so that the strong converse does not hold (one may use Examples 1-3 to construct component sequences $\bW, \bZ$ for further insights). Note that this conclusion holds for any $p$ as long as $0< p<1$.

\begin{remark}
It is tempting to conclude, based on $\uuline{H}(\bXi) = \min\{\uuline{H}(\bW),\uuline{H}(\bZ)\}$ which holds in full generality, that \eqref{eq.ulnH(bXi)} should hold with equality in general. To see that this is not the case, consider Example 4 with the following component sequences:
\ba\notag
w_s^n &=\{b_1..b_s,0..0\}\\
z_s^n &=\{0..0,b_{s+1}..b_n\}
\ea
where $b^n$ is a binary i.i.d. equiprobable sequence. This models a scenario where the noise randomly corrupts either 1st or 2nd part of a codeword and $s$ controls its length. It follows that \ba
\ooline{H}(\bXi)= \uln{H}(\bW)=\uln{H}(\bZ)=1
\ea
yet
\ba
\uln{H}(\bXi)=1/2 < 1 = \min\{\uln{H}(\bW),\uln{H}(\bZ)\}
\ea
Note that the strong converse does not hold in this case either, even though it holds for each component sequence individually and $\uln{H}(\bW)=\uln{H}(\bZ)$. Further note that $\oln{H}(\bXi)=1/2$, $\inf_s \oln{H}(\bXi|s) = \sup_s \oln{H}(\bXi|s) =1$ so that the last inequality in \eqref{eq.cor.sc} does not hold.
\end{remark}

\section{Conclusion}
The capacity of compound channels with additive noise and the Rx CSI has been studied. When all alphabets are discrete and there is no cost constraint, noiseless causal feedback does not increase the capacity. The impact of the channel state information at the transmitter has been quantified. In particular, it does not increase the capacity if the additive noise is a uniform compound process. Otherwise, it may provide significant improvement (unlike the feedback), which was shown via examples. A saddle-point has been shown to exist  in the information-theoretic game between the transmitter and the nature, even though the objective is not convex/concave in the right way. Finally, the sufficient and necessary condition for the strong converse to hold has been establish: it requires the worst-case noise sequence to behave ergodically, even though no ergodicity or information satiability requirements were imposed upfront. Examples are provided to facilitate understanding and insights.

\section{Acknowledgement}

The authors are grateful to A. Lapidoth and E. Telatar for their support, and to P. Mitran and M. Raginsky for fruitful discussions.

\section{Appendix}

\subsection{Proof of Lemma \ref{lemma.ooH-supHs}}
The proof of the 1st part (the inequality in general) is by contradiction. Assume that $\ooline{H}(\bXi) < \sup_s \overline{H}(\bXi|s)$, which implies that
\ba
\exists s_0 : \ooline{H} = \ooline{H}(\bXi) < \overline{H} = \overline{H}(\bXi|s_0)
\ea
Set
\ba
R = (\ooline{H} + \overline{H})/2 = \ooline{H} + \Delta = \overline{H} - \Delta
\ea
where $ \Delta = (\overline{H} - \ooline{H})/2 > 0$. Note that
\ba
\lim_{n\rightarrow\infty} \Pr\left\{\frac{1}{n} h(\Xi^n_{s_0}|s_0) > \overline{H} - \Delta \right\} > 0
\ea
from the definition of $\overline{H}$. However,
\ba\notag
0 &= \lim_{n\rightarrow\infty} \sup_s \Pr\left\{\frac{1}{n} h(\Xi^n_{s}|s) > \ooline{H} + \Delta \right\}\\
&\ge \lim_{n\rightarrow\infty} \Pr\left\{\frac{1}{n} h(\Xi^n_{s_0}|s_0) > \ooline{H} + \Delta \right\}\\ \notag
&= \lim_{n\rightarrow\infty} \Pr\left\{\frac{1}{n} h(\Xi^n_{s_0}|s_0) > \overline{H} - \Delta \right\} >0
\ea
where 1st equality is due to the definition of $\ooline{H}$, i.e. a contradiction,  from which the desired inequality follows.

The "if" part of the equality case (under uniform noise) is also proved by contradiction: assume that, under the uniform convergence,
\ba
\ooline{H} > \overline{H} = \sup_s \overline{H}(\bXi|s)
\ea
and set
\ba
R = (\ooline{H} + \overline{H})/2 = \ooline{H} - \Delta = \overline{H} + \Delta
\ea
where $ \Delta = (\ooline{H}- \overline{H})/2 > 0$, and hence
\ba
\lim_{n\rightarrow\infty} \Pr\left\{n^{-1}  h(\Xi^n_{s}|s) > \overline{H} + \Delta \right\} = 0\ \forall s \in \sS
\ea
from the definition of $\overline{H}$, so that a contradiction follows
\ba\notag
0 &= \sup_s \lim_{n\rightarrow\infty} \Pr\left\{n^{-1}  h(\Xi^n_{s}|s) > \overline{H} + \Delta \right\}\\
&= \lim_{n\rightarrow\infty} \sup_s \Pr\left\{n^{-1}  h(\Xi^n_{s}|s) > \overline{H} + \Delta \right\}\\ \notag
&= \lim_{n\rightarrow\infty} \sup_s \Pr\left\{n^{-1}  h(\Xi^n_{s}|s) > \ooline{H} - \Delta \right\} > 0
\ea
where 2nd equality is due to uniform convergence and the last inequality is from the definition of $\ooline{H}$.

To prove the "only if" part, assume that the equality holds and observe that
\ba\notag
0 &= \lim_{n\rightarrow\infty}\sup_s \Pr\left\{n^{-1}  h(\Xi^n_{s}|s) > \ooline{H} + \Delta \right\}\\
&= \lim_{n\rightarrow\infty}\sup_s \Pr\left\{n^{-1}  h(\Xi^n_{s}|s) > \sup_s \overline{H}(\bXi|s) + \Delta \right\}
\ea
for any $\Delta>0$. The last equality implies uniform convergence: for any $\epsilon>0$ there exists such $n_0(\epsilon)$ that for any $n>n_0(\epsilon)$,
\ba\notag
\sup_s \Pr\left\{n^{-1} h(\Xi^n_{s}|s) > \sup_s \overline{H}(\bXi|s) + \Delta \right\} < \epsilon
\ea
and hence the convergence is uniform.

\subsection{Proof of Proposition \ref{prop.sc.prop}}

Let $\liminf= \lim_{n\rightarrow\infty} \inf_{s}$ and likewise for $\limsup$. Eq. \eqref{eq.sc.P1} follows from the definition of $\uln{\{\cdot\}}$:
\ba
\notag
\uln{(-\bX)} &=  \sup \bLF x: \liminf \Pr\left\{ -X_{sn} \le x \right\} =0 \bRF\\ \notag
    &=  \sup \bLF x: \liminf \Pr\left\{ X_{sn} \ge -x \right\} =0 \bRF\\ \notag
    &=  \sup \bLF -z: \liminf \Pr\left\{ X_{sn} \ge z \right\} =0 \bRF\\ \notag
    &=  -\inf \bLF z: \liminf \Pr\left\{ X_{sn} \ge z \right\} =0 \bRF\\
    &=-\oln{\bX}
\ea
To prove \eqref{eq.sc.P2}, set $x=\oln{\bX}+\ooline{\bY} +\delta$ for some $\delta>0$, let $B$ denote the event $\{Y_{sn}<\ooline{\bY}+ \delta\}$ and $B^c$ - its complement, and observe that
\ba
\label{eq.sc.pp.2}
\notag
0 &= \liminf \Pr\{X_{sn}+\ooline{\bY} \ge x\}\\ \notag
    &= \liminf (\Pr\{X_{sn}+\ooline{\bY} \ge x| B\}\Pr\{B\}+ \Pr\{X_{sn}+\ooline{\bY} \ge x|B^c\}\Pr\{B^c\}) \\ \notag
    &\ge \liminf \Pr\{X_{sn}+\ooline{\bY} \ge x| B\}\Pr\{B\} \\ \notag
    &\ge \liminf \Pr\{X_{sn}+ Y_{sn} -\delta \ge x| B\}\Pr\{B\} \\ \notag
    &= \liminf \Pr\{X_{sn}+ Y_{sn} -\delta \ge x| B\}\Pr\{B\} + \limsup \Pr\{X_{sn}+ Y_{sn} -\delta \ge x| B^c\}\Pr\{B^c\}\\ \notag
    &\ge \liminf (\Pr\{X_{sn}+ Y_{sn} -\delta \ge x| B\}\Pr\{B\} + \Pr\{X_{sn}+ Y_{sn} -\delta \ge x| B^c\}\Pr\{B^c\})\\
    &= \liminf \Pr\{X_{sn}+ Y_{sn} \ge x +\delta \} =0
\ea
where 1st equality is from $x=\oln{\bX}+\ooline{\bY} +\delta$ and the definition of $\oln{\bX}$; 2nd inequality is from $\ooline{\bY}> Y_{sn}- \delta$ conditioned on $B$; 3rd equality is from
\ba
\limsup \Pr\{X_{sn}+ Y_{sn} -\delta\} \ge x| B^c\}\Pr\{B^c\} \le \limsup \Pr\{B^c\} = 0
\ea
where the last equality is from the definition of $B^c$; the last equality in \eqref{eq.sc.pp.2} is implied by the preceding chain. This last equality implies that $\oln{\bX+\bY} \le x +\delta$ so that
\ba
\oln{\bX+\bY} \le \oln{\bX}+\ooline{\bY} + 2\delta
\ea
for any $\delta>0$, which proves 2nd inequality in \eqref{eq.sc.P2}. To prove 1st one, use the substitutions $\bY \rightarrow -\bY$ and $\bX \rightarrow \bX+\bY$ in combination with \eqref{eq.sc.P1}.

To establish \eqref{eq.sc.P3}, we first show that $\uuline{\bX} \le \uln{\bX}$. To this end, let
\ba \notag
\Omega_1 &=  \bLF x: \limsup \Pr\left\{ X_{sn} \le x \right\} =0 \bRF\\
\Omega_2 &=  \bLF x: \liminf \Pr\left\{ X_{sn} \le x \right\} =0 \bRF
\ea
Since
\ba
\limsup \Pr\left\{ X_{sn} \le x \right\} \ge \liminf \Pr\left\{ X_{sn} \le x \right\}
\ea
it follows that $\Omega_1 \in \Omega_2$, which implies $\uuline{\bX} \le \uln{\bX}$ by using $\sup$. Next, we show that $\uuline{\bX} \le \oln{\bX}$. To this end, let
\ba \notag
\Omega_3 &=  \bLF x: \limsup \Pr\left\{ X_{sn} \le x \right\} =1 \bRF
\ea
and observe that
\ba \notag
\oln{\bX} &= \inf \bLF x: \liminf \Pr\left\{X_{sn} \ge x \right\} =0 \bRF \\
    &= \inf \bLF x: \liminf \Pr\left\{ X_{sn} > x \right\} =0 \bRF\\ \notag
    &= \inf \bLF x \in \Omega_3 \bRF
\ea
Since, for any $x_1\in \Omega_1$ and any $x_3\in \Omega_3$, it holds that $x_1< x_3$, so that
\ba
\uuline{\bX} = \sup \bLF x \in \Omega_1 \bRF \le \inf \bLF x \in \Omega_3 \bRF = \oln{\bX}
\ea
This establishes 1st inequality in \eqref{eq.sc.P3}. 2nd one is trivial. 3rd one can be established from 1st one using $\bX \rightarrow - \bX$.

To show 1st inequality in \eqref{eq.sc.P4}, recall that
\ba
\uln{\bX}_s &= \sup \bLF x: \lim_{n\rightarrow\infty} \Pr\left\{X_{sn} \le x \right\} =0 \bRF,
\ea
set $x_0=\uln{\bX}_s -\delta$ for some $\delta>0$ and observe that
\ba
0 = \lim_{n\rightarrow\infty} \Pr\left\{X_{sn} \le x_0 \right\} \ge \liminf \Pr\left\{X_{sn} \le x_0 \right\} = 0
\ea
where the last equality is implied by the preceding chain. This implies that $\uln{\bX} \ge x_0$. Since this holds for any $\delta>0$, $\uln{\bX} \ge \uln{\bX}_s$ follows. Since this holds for any $s$, 1st inequality in \eqref{eq.sc.P4} follows. 2nd one can be established via $\bX \rightarrow -\bX$.

To establish \eqref{eq.sc.P5}, observe that $Y_{sn} \overset{\Pr} {\rightarrow} y_0$ implies $\uuline{\bY} = \ooline{\bY}=y_0$ and use \eqref{eq.sc.P2}.

\subsection{Proof of Theorem \ref{thm.sc}}
We begin with a brief summary of the sufficient and necessary condition for the general compound channel to satisfy the strong converse.

\begin{thm}[\cite{Loyka-15-2}\cite{Loyka-16}]
The general compound channel with full Rx CSI and without feedback satisfies the strong converse condition if and only if
\ba
\label{eq.strong.conv.1}
C_c \triangleq \sup_{\bX} \uuline{I}(\bX;\bY) = \sup_{\bX} \oln{I}(\bX;\bY)
\ea
where $\sup$ is over all sequences of finite-dimensional input distributions. The condition \eqref{eq.strong.conv.1} is equivalent to the following: for any $\delta>0$ and an optimal input $\bX^*$ ,
\ba
\label{eq.strong.conv.2}
\lim_{n\rightarrow\infty} \inf_s \Pr\{|Z_{ns}^* - C_c|> \delta\} = 0
\ea
where $Z_{ns}^*=\frac{1}{n} i({X^n}^*;{Y^n}^*|s)$ is the information density rate under optimal input $\bX^*$, i.e. there exists such sequence of channel states $s(n)$ that the corresponding information density rate $Z_{ns}^*$ under optimal input $\bX^*$ converges in probability to the compound capacity $C_c$ (i.e. the channel represented by this sequence of states is information-stable, even though the original compound channel is not required to be information-stable).
\end{thm}

To adapt this result to the feedback case, we again consider $\bW$ as an input and optimize over both $\bW$ and $\bF$ so that \eqref{eq.strong.conv.1} becomes
\ba
\label{eq.strong.conv.1FB}
\sup_{\bW,\bF} \uuline{I}(\bW;\bY) = \sup_{\bW,\bF} \oln{I}(\bW;\bY)
\ea
Since the left-hand side has been already evaluated, we now evaluate the right-hand side. To this end, one can follow the steps similar to those in evaluating the left-hand side. First, observe that
\ba\notag
\sup_{\bW,\bF} \oln{I}(\bW;\bY) &\le \sup_{\bW,\bF} [\ooline{H}(\bY) - \uln{H}(\bY|\bW)]\\ \notag
&\le \log M - \inf_{\bW,\bF}\uln{H}(\bY|\bW)\\
&= \log M - \uln{H}(\bXi)
\ea
where 1st inequality is due to Proposition \ref{prop.sc.prop}; 2nd inequality follows from $\ooline{H}(\bY) \le \log M$ (since the alphabet is $M$-ary); the last equality is due to \eqref{eq.p(sy)=p(xi)} so that $\uln{H}(\bY|\bW)=\uln{H}(\bXi)$. Now, using no feedback and uniform input $\bX$, one obtains $\oln{I}(\bW;\bY) = \log M - \uln{H}(\bXi)$ so that
\ba
\sup_{\bW,\bF} \oln{I}(\bW;\bY) \ge \log M - \uln{H}(\bXi)
\ea
Combining the two inequalities,
\ba
\sup_{\bW,\bF} \oln{I}(\bW;\bY) = \log M - \uln{H}(\bXi)
\ea
It is remarkable that, similarly to $\uuline{I}(\bW;\bY)$, the optimal value of $\oln{I}(\bW;\bY)$ is not affected by feedback either and the best strategy is to use the uniformly-distributed input and ignore feedback. Combining the last equality with \eqref{eq.C_FB=C_NFB}, the desired condition follows.

\end{document}